\documentclass[letterpaper,twocolumn,10pt]{article}
\usepackage{usenix,epsfig,endnotes,xspace}
\usepackage{amsthm,amsmath,amssymb}
\usepackage{epigraph}
\usepackage[draft]{todonotes}
\usepackage{algpseudocode}
\usepackage{graphicx}
\usepackage{subcaption}
\usepackage{url}
\newcommand{\sysname}{SybilQuorum\xspace}
\newcommand{\sysnamec}{\sysname-core\xspace}
\newcommand{\sysnameh}{\sysname-hybrid\xspace}

\newcommand{\para}[1]{\vspace{3mm} \noindent {\bf #1.}}

\newtheorem{theorem}{Security Theorem}
\newtheorem{definition}{Definition}
\newtheorem{lemma}{Lemma}

\begin{document}

\date{}

\title{\Large \bf \sysname: Open Distributed Ledgers Through Trust Networks \\ (Extended Abstract for the Stanford Blockchain Conference)}

\author{
{\rm Alberto Sonnino\thanks{This work was done while the authors where at \texttt{chainspace.io}.}}\\
University College London
\and
{\rm George Danezis\footnotemark[1]}\\
University College London
}

\maketitle


\subsection*{Abstract}

The Sybil attack plagues all peer-to-peer systems, and modern open distributed ledgers employ a number of tactics to prevent it from proof of work, or other resources such as space, stake or memory, to traditional admission control in permissioned settings. With \sysname we propose an alternative approach to securing an open distributed ledger against Sybil attacks, and ensuring consensus amongst honest participants, leveraging social network based Sybil defences. We show how nodes expressing their trust relationships through the ledger can bootstrap and operate a value system, and general transaction system, and how Sybil attacks are thwarted. We empirically evaluate our system as a secure Federated Byzantine Agreement System, and extend the theory of those systems to do so.

\section{Introduction}



Distributed ledgers, and blockchains, as they are sometimes called, provide peer-to-peer open transaction systems used for alternative currencies, such as Bitcoin~\cite{nakamoto2008bitcoin}, or general distributed execution of code, often called `smart contracts'~\cite{wood2014ethereum}. The main innovation of Nakamoto consensus~\cite{nakamoto2008bitcoin}, underpinning both systems, is the open nature of the system that allows infrastructure nodes to come-and-go, and participate on the basis of solving proof-of-work cryptographic puzzles. However, this is computationally expensive and resource intensive. 

Alternatives based on proof-of-stake~\cite{dai1998b} do not consume resources, but require nodes to lock some `stake' in a native crypto-currency, to participate and slash this stake upon detecting misbehaviour. However, there are valid concerns around such systems: locked stake represents a loss of opportunity. The consensus favors `richer' nodes, that as a result get richer, which may in turn threaten decentralization and may lead to attacks by minority players with wealth. Finally, the values it embeds relate to `boardroom democracy' (as Bryan Ford suggests), and may not be aligned with principles of openness and equity.

In this work we introduce a new consensus mechanism, \sysname, that allows peers to establish a distributed ledger without the need for either proof-of-work, or other physical resources, or proof-of-stake to eliminate Sybil attacks. The system is open to new members, and permissionless, making it competitive with Nakamoto consensus. It is based on an established line of work related to Sybil defences leveraging Social Networks, starting with protecting Distributed Hash Tables~\cite{DBLP:conf/esorics/DanezisLKA05}, and pursued by SybilGuard~\cite{DBLP:journals/ton/YuKGF08}, SybilLimit~\cite{DBLP:journals/ton/YuGKX10}, and SybilInfer~\cite{DBLP:conf/ndss/DanezisM09}. We also present a hybrid system, \sysnameh, that combines stake and social networks to further strengthen Sybil resistance.

Our contributions include: 
\begin{enumerate}
    \item A proposal for achieving open consensus backed by social links, embodied in the \sysnamec design.
    \item Extensions to integrate aspects of proof-of-stake to enhance Sybil defences, and prevent wealthy nodes from controlling the consensus, namely \sysnameh. 
    \item An extension to the theory of Federated Byzantine Agreement Systems (FBAS), and efficient algorithms based on this theory to test for their safety and liveness. 
    \item A concrete design, including integration with specific consensus mechanisms compatible with \sysname at a systems level.
    \item An evaluation of the strength of \sysname based on real-world social graphs, and the conditions under which it enhances security against Sybil attacks.
\end{enumerate}

\section{Background and Related work}

The Sybil attack was introduced by Douceur~\cite{DBLP:conf/iptps/Douceur02}, in relation to engineering peer-to-peer systems, and identified types of defences: admission control through central authentication, and resource constraints. Permissioned ledgers, such as Hyperledger~\cite{DBLP:conf/eurosys/AndroulakiBBCCC18} or Quora, take the first approach, and only allow known and designated nodes to participate in consensus. Open distributed ledgers, including Bitcoin and Ethereum follow the second paradigm. Proof of Work was first proposed by Back, as Hashcash~\cite{back2002hashcash}, to prevent Denial of Service. In the context of spam its economic efficiency was questioned by Clayton and Laurie~\cite{laurie2004proof}.

Proof-of-stake systems were proposed first in the 90s by Wei Dai, in B-money~\cite{dai1998b}. Modern proof-of-stake systems, such as Ouroboros~\cite{kiayias2017ouroboros} allow users to lock and delegate stake, and sample those users proportionately to their stake to determine an order in which blocks are produced in a blockchain system. Consensus therefore remains open, in that anyone who can buy some currency and lock it as stake can participate. However, there are serious concerns with this approach: the most fundamental one being that very wealthy parties may afford to acquire a lot of stake, and abuse it to extract value out of the system. Since stake often allows nodes to mine blocks, and reap rewards, the economics of proof-of-stake may lead to oligarchies through a ``rich get richer'' dynamic.

A number of blockchain systems consider some trust judgments between nodes and leverage them to achieve consensus. Stellar~\cite{mazieres2015stellar} considers that each node links to other nodes, and uses these direct trust judgments to form quorums in which byzantine consensus may be run --- this is the closest related system to \sysname. We use the definition and security concepts introduced in Stellar, such as a Federated Byzantine Agreement System (FBAS) and a disposable set (DSet), as a basis for our security arguments and evaluation. However, Stellar only considers direct judgments to form an FBAS, rather than the topology of the full social network.
Ripple~\cite{armknecht2015ripple} allows participants to connect to each other, but does not solve the Sybil defence problem directly, and does not achieve inter-node consensus in a strong manner. Instead, each node may run its own currency and economy, and rely on others' willingness to act as an exchange to transfer value between nodes that are not directly connected.

Besides blockchains, systems leveraging social networks --- and explicit trust judgments of users about each other --- have been proposed to combat Sybil attacks. Early work considers leveraging the `introduction graph', by which nodes get to access a Distributed Hash Table through other nodes, to ensure routing security~\cite{DBLP:conf/esorics/DanezisLKA05}. Raph Levien productized those ideas to extract reputation of developers in `Advogato'~\cite{levien2009attack}; and Sam 
Lessin~\cite{samlessin} proposed using trust graphs backed by financial commitments to infer the financial trustworthiness of users in a graph in the context of blockchains.

Academic works within this family of systems consider general social network information distributed in a peer to peer network to allow each node to determine which other nodes are genuine or Sybils. In this line of work SybilGuard~\cite{DBLP:journals/ton/YuKGF08} and SybilLimit~\cite{DBLP:journals/ton/YuGKX10} perform a distributed computation, using random walks in a network, to determine the honest regions within it. SybilInfer~\cite{DBLP:conf/ndss/DanezisM09} takes a centralized approach, and analyzes a stored social graph to identify potential Sybil regions. 

These defences make some security assumptions related to the topology of `honest' social graphs: those need to be fast mixing, have small diameter, and contain relatively few links to nodes being part of a Sybil attack. Those systems allow each node to extract a degree of belief about whether any other node in the system is a genuine participant or a Sybil. However, this belief depends on the position of the node in the social graph, and may not be exactly the same even for two honest nodes --- thus they do not directly lead to any form of consensus, not even about who is a Sybil node. Ultimately, each node uses those degrees of belief to define their own set of nodes considered honest.

Subsequent work questions a number of assumptions based on the analysis of real-world social graphs~\cite{DBLP:conf/ccs/MohaisenTHK12}. This work is influential in that it highlights that the social graphs on which these defences rest, but truly capture trust judgments, and provide incentives for users to not accept any links, including to malicious nodes. In this work we also highlight a further limit of SybilInfer as originally proposed: it is an effective mechanism to detect Sybil regions in the presence of an attack, however it is also presenting a large number of ``false positives'' when the network is free of such attacks --- by misclasifying a large number of honest nodes as Sybils. We provide a solution to this problem.

Besides `blockchain' based consensus, based on a chain of blocks and a fork choice rule, modern distributed ledgers consider and reimagine more traditional forms of byzantine consensus. An exemplary system is Tendermint~\cite{kwon2014tendermint}, that combines a quorum based byzantine consensus protocol, with a proof-of-stake mechanism. In Tenderming, and in general, decisions are made as part of the consensus protocol when over two-thirds of `stake' supports a decision --- abstracting from the actual identities of nodes and only considering their weight in stake. The advantages of this approach is low latency, quick finality, and higher throughput than Nakamoto consensus. This family of systems also includes Blockmania~\cite{danezis2018blockmania}, which separates messages materialized and exchanged in a network forming a directed acyclic graph of blocks, from the process of nodes independently interpreting it to reach consensus and order transactions. This separation is key for the practical and efficient implementation of \sysname.

\section{The \sysnamec system}

The \sysnamec system is the purest instantiation of the ideas behind \sysname. It maintains a distributed ledger, including a social network of user trust judgments about each other. In turn it leverages this information to maintain the network consistent across honest users, and also to order arbitrary transactions. Its key security property is that two nodes will accept the same sequence of transactions if they are sufficiently related in the social graph, and sufficiently separated from Sybil nodes. Those may then be used to implement any distributed computation following the well established state-machine replication paradigm~\cite{schneider1990implementing}.

\para{Security State of the Ledger} Each user maintain a local copy of the ledger, which consists of two types of information: security related information, and a sequence of application transactions. The security information relates to \sysname operations, while the transactions can be arbitrary and are never interpreted by \sysname. Specifically, the security information consists of a set of directed \emph{links} between users. A user Alice, represented by a public key, may sign a statement that she trusts a user Bob, by public key, to not be a Sybil: this becomes an arc between Alice's public key and Bob's, denoted as $pk_A \rightarrow pk_B$. Alice may also sign a statement removing such an arc -- a sequence number prevents replay attacks in either adding or removing such arcs. Signed statements adding or removing links are processed through the consensus protocol, and accepted (or not) by nodes in the network --- in a manner we will shortly examine. When accepted the security information is updated to reflect the new social graph. 

Ultimately at any moment each node has a representation of the security state of the ledger, namely a directed graph of links between public keys.

\para{Sybil Defences \& the Security State} Upon every update of the security state of the ledger a node performs an analysis of the latest social graph to determine the probability with which each node may be controlled by a Sybil attacker. Applying techniques from SybilInfer~\cite{DBLP:conf/ndss/DanezisM09} the outcome of the analysis for node $v_i$ is a map between public keys of nodes $pk_j$ and a weight $w_{ij} \in [0, 1]$ representing the probability the node is honest, represented as $pk_i \rightarrow w_{ij}$.

We note this is \emph{local judgment}: the node $v_i$ may ascribe a different weight to $pk_j$, than node $v_{i'}$, namely in general we expect $w_{ij} \neq w_{i'j}$. This is the case even if all nodes involved, namely $v_i$, $v_{i'}$, and $v_j$ are all honest. Thus the map, even between two honest nodes cannot be assumed to be the same --- and Sybil defences by themselves cannot in general achieve consensus; not even on who is a Sybil.

Despite local judgments being different in their exact details, we do not expect them to be uncorrelated. Since at their core Sybil defence mechanisms applied by honest nodes will tend to ascribe higher probability to nodes that are honest, from those nodes that are actually part of a large Sybil attack. Therefore we assume that the lists of nodes that two honest nodes will extract from the Sybil defence mechanism are going to be largely composed of honest nodes, and also likely to be overlapping.

\para{From SybilInfer weights to presumed honest sets} Our experiments with SybilInfer uncover a shortcomming of the system as originally proposed. In the presence of a Sybil attack, it is effective at detecting it --- namely setting the weights $w_{ij}$ as larger than $y = 1/2$ for honest nodes, and lower for Sybil nodes. Therefore each node in \sysname may select a set of other nodes to consider as honest according to the criterium $w_{ij} \geq y$.

However, its probabilistic model is calibrated assuming there is an attack, and in the absence of a large Sybil region it misclassifies a significant number of honest nodes as Sybils. Therefore we need to set a dynamic threshold $y$ that is sensitive to whether a cut in the graph is the result of a Sybil attack, or `natural' given a social network.

Our mechanism for calibrating the cut-off $y$ is based on the fundamental insights behind social network based Sybil defences: we consider a region of the graph as being composed of Sybils, if the volume of links to this region are comparatively low, compared with the size of the honest graph, and in particular the number of links within this honest sub-graph. We define the node set $H_y \subseteq V$ containing all nodes $v_j$ with $w_{ij} \geq y$, and $S_y \subseteq V$ with all the nodes $v_j$ such that $w_{ij} < y$. We also define a function $\mathcal{L}(N_0, N_1)$, over a set of nodes $N_0, N_1 \subseteq V$, that represent the number of unidirectional links between the node sets $N_0$ and $N_1$. We select the largest cut-off value $y \in [0.45, 0.55]$ such that $\mathcal{L}(H_y, H_y) > \mathcal{L}(N_y, S_y)$. We then use the selected value of $y$ to define for each nodes the set of honest nodes $H(v_i) = \{ v_j \, | \, w_{ij} \geq y \}$.

Intuitively this selects a cut-off $y$ that ensures that the number of links to the Sybil region is indeed small, and in particular smaller than the number of links within the honest region --- a sign of an actual Sybil attack. Large cuts in honest networks, will have a very large number of links between the two honest regions. Since in a largely honest graph the difference in weights $w_{ij}$ is due to the nodes $v_i$ proximity to some nodes, more than others, the actual number of links between the regions will be large, and we would not select such a cut as a Sybil attack. We validate this approach through experiments on real graphs.

\para{From social network Sybil defences to consensus} We have already highlighted that the Sybil defences alone do not lead to any sort of consensus between honest nodes. However, we can leverage them, and the assumptions we make about them to achieve consensus. To achieve consensus we use the definitions, safety and liveness conditions determined by the Stellar~\cite{mazieres2015stellar} protocol for Federated Byzantine Agreement Systems (FBAS). 

\begin{definition}[FBAS]
A Federated Byzantine Agreement System, or \emph{FBAS}, is a pair $\langle V, Q \rangle$ consisting of a set of nodes $V$ and a quorum function $Q\ :\ V \rightarrow 2^{2^V} \setminus \varnothing$  specifying one of more quorum slices for each node, where a node belongs to all of its own quorum slices---i.e.\ $\forall v \in V, \forall q \in Q(v), v \in q$. (Note $2^X$ denotes the powerset of $X$.) (From~\cite{mazieres2015stellar}.)
\end{definition}

\begin{definition}[Quorum]
A set of nodes $ U \subseteq V$ in FBAS $\langle V, Q \rangle$ is a \emph{quorum} iff $U \neq \varnothing$ and $U$ contains a slice for each member---i.e.\ $\forall v \in U, \exists q \in Q(v)$ such that $q \subseteq U$. (From~~\cite{mazieres2015stellar}.)
\end{definition}

We leverage \sysname to create an FBAS in the folowing manner. The set of all nodes $V$ includes all nodes $v_i$ in the system, honest and Sybils. Each honest node uses the social network Sybil defence mechanism, to define a list of nodes $H(v_i)$ that it considers honest. It does so by including in $H(v_i)$ all nodes $i$ such that $w_{ij} > y$, where $y$ is the selected cutoff value in $[0, 1]$ (see above for how to select $y$). The quorum slices for each node $Q(v_i)$ are all the subset of $H(v_i)$ of cardinality greater than $2/3 | H(v_i) | $.


We discuss in our evaluation what it means for such an FBAS system to be secure, and also provide a theory for how to efficiently test an FBAS is secure.

\section{The \sysnameh extensions}

Previous work~\cite{DBLP:conf/ccs/MohaisenTHK12} argues that `natural' social graphs do not provide the fast mixing properties necessary for reliably detecting Sybil attacks. Furthermore, research suggests that, at least some, users are easily defrauded into connecting on social network platforms to other users without much due diligence as to the identity or trustworthiness of the profile.

We extend the \sysname system, and describe \sysnameh, that combines it with aspects of proof-of-stake for two purposes: (1) as traditional proof-of-stake it caps the ability of dishonest nodes to create an infinity of identities at will; and (2) it provides incentives for honest nodes to be careful when connecting to other nodes, and potential penalties for making poor judgments. The first property keeps the number of potential Sybil nodes low, while the second one supports the key property necessary in social network Sybil defences namely that the capacitance of the graph from the honest region to the dishonest region remains small.

\para{Token system} All proof-of-stake systems require a token system, with a fixed supply, to be integrated within the security state of a ledger. Nodes may then `lock' tokens as `stake'. \sysnameh also requires such a token system, and it may be abstracted as a map from accounts (as public keys) to token values, namely $pk_i \rightarrow v_i$, maintained by all nodes. We consider those tokens are not forgeable and are fungible, as per other crypto-currencies.

\para{Links with stake} \sysnameh allows nodes to create arcs to other nodes representing judgments about their trustworthiness, as part of the security state of the system. However, unlike \sysnamec, those arcs are associated with a value in tokens. Such an arc from Alice to Bob, with value $v_{AB}$ is denoted as $pk_A \xrightarrow{v_{AB}} pk_B $. The transaction creating those arcs is signed by the originator Alice, using $pk_A$, and the value $v_{AB}$ is deducted from her account. Thus the security state of \sysnameh consists of a directed weighted graph between nodes.

\sysnameh enables not only the originators or arcs to remove them, and recuperate the value assigned to them, but also the destination of arcs. Therefore if an arc  $pk_A \xrightarrow{v_{AB}} pk_B $ exists in the system, Bob may issue a transaction signed by $pk_B$ to remove the arc and increase his balance by $v_{AB}$. We allow both sides to reclaim the value of a link in order to ensure it is a reliable signal to others of the trust between nodes. Alice, by creating a link with value $v_{AB}$ to Bob, trusts him to not immediately or eventually `steal' this value. Self-imposed \emph{vulnerability implies trust}, and signals it in a way that can be relied upon by other nodes.

This vulnerability also penalizes honest nodes that may be more likely to make poor trust decisions that degrade the overall social network Sybil defence mechanism --- assuming there is a different propensity amongst honest nodes to make poor decisions. Such nodes will pick dishonest nodes to make arcs to more often. Many of those dishonest nodes will not be part of a Sybil attack, but merely fraudsters that will simply reclaim the value. As a result bulk dishonesty, protects the system from strategic dishonesty --- honest nodes with poor judgment are likely to be disincentivized from creating links, and impoverished if they continue doing so recklessly. 

\para{Weighted social network defences} As soon as a \sysnameh node updates the security state, and in particular the weighted directed graph representing the social network, it re-runs a Sybil Detection algorithm. However, traditional social network Sybil detection algorithms, such as SybilInfer do not operate on weighted graphs --- and require some modifications to operate.

We first prune the social graph from all arcs that are not reciprocated namely all arcs connecting two nodes $pk_i \rightarrow pk_j$ for which there is no arc $pk_j \rightarrow pk_i$. We define the total value committed in remaining links by a node as $V_i = \sum_{x \in N} v_{ix}$, where $v_{ix}$ is the value assigned by $n_i$ to arcs to each node $n_x$ (by convention we consider that if an arc does not exist its value is zero).

Our goal is to then define a Markov-chain over the nodes $n_i$ with a stationary distribution $\pi(n_i) = \frac{V_i}{\sum_{n \in N} V_n}$, namely one proportional to the amount of `stake' each node has committed to reciprocal links. (SybilInfer targets instead a uniform stationary distribution).

Any distribution $g(n_i | n_j)$ that maintains `detailed balance' would ensure this property, namely $p(n_i | n_j) \cdot \pi(n_j) = p(n_j | n_i) \cdot \pi(n_i)$. However we wish to limit the transition matrix of the chain to only have support on reciprocal arcs. We therefore define a proposal distribution for node $n_i$ as $g(n_j | n_i) = v_{ij} / V_i$. This proposal is accepted (following the Metropolis-Hasting MCMC method) with probability:
\begin{align}
\alpha_{ij} &= \min \left\{1, \frac{\pi(n_j)}{\pi(n_i)} \cdot \frac{g(n_i | n_j)}{g(n_j | n_i)} \right\}    \\
       &= \min \left\{1, \frac{V_j / \sum_n V_n}{V_i / \sum_n V_n} \cdot \frac{w_{ji} / V_j}{w_{ij} / V_i} \right\} \\
       &= \min \left\{1, \frac{w_{ji}}{w_{ij}} \right\}
\end{align}

\noindent Thus the transition matrix becomes: 
\begin{align}
    p(n_j | n_i) &= \alpha_{ij} \cdot w_{ij} / V_i \\
    &= \min \left\{\frac{w_{ij}}{V_i}, \frac{w_{ji}}{V_i} \right\}
\end{align}
with the remaining probability mass being assigned to the self-transition $p(n_i | n_i)$. Interestingly inter-node transitions are only influenced by the lower value of $w_{ij}$ and $w_{ji}$,

An walk on this Markov chain results at a node drawn from the stationary distribution $\pi$, as its length tends towards infinity. However we want to leverage the properties of short random walk on such graphs: a path of length $\ell = m \cdot \log N$ should converse towards $\pi$, but would be disrupted by the capacitance between the honest and Sybil region in the graph. As a result short random walks starting at an honest node will tend to remain within the honest sub-graph. We denote the distribution of nodes reached after such a short walk starting at an honest node $n_i$ by $\pi_i^*$. Mathematically this means that for honest nodes $h \in N. \pi_i^*(n_h) \geq V_h / \sum_n V_n$ and conversely for Sybil nodes $s \in N. \pi_i^*(n_s) \leq V_s / \sum_n V_n$.

We leverage this `gap' and amplify it to penalize nodes that are more likely to be Sybils, using the Logistic function:
\begin{equation}
    \text{Logistic}(x, x_0, k) = \frac{1}{1 + e^{-k(x - x0)}}
\end{equation}
Each honest node $v_i$ assigns a probability to other nodes $v_j$ being honest, computed as:
\begin{equation}
    w_{ij} = \text{Logistic}(\pi_i^*(n_j), \frac{V_j}{\sum_n V_n}, k)
\end{equation}
The logistic term takes values in $[0, 1]$, with values closer to zero if $\pi^*_i$ undershoots the target $\pi$, and closer to 1 if it overshoots it. The term then scales the stake of the node $V_j$, allowing it to take closer to its maximal value for honest nodes; and becoming closer to zero for dishonest nodes. The parameter $k$ is a measure of how sharply deviations lead to minimum of maximum values, and can be chosen by the nodes depending on their topology and connectivity into the social graph.

\para{Defining an FBAS} As for \sysnamec, we define an FBAS $\langle V, Q \rangle$, by having each node $v_i \in V$ using the resulting weights $w_{ij}$ and considering a set of nodes $H(v_i)$ honest if those weights are at least a cutoff value $ w_{ij} > y$. That cut-off value is selected as in \sysnamec to prevent large numbers of false positives. The quorum function $Q(v_i)$ for each node $v_i$, contains all subsets of $H(v_i)$ of size greater than $2/3 |H(v_i)|$.

\section{The Security of an FBAS}

Both proposed variants of \sysname define a Federated Byzantine Agreement System, as defined in~\cite{mazieres2015stellar}. We therefore use some further definitions to achieve two security properties: (1) safety means that two honest nodes will agree to the same outcome of the consensus; and (2) liveness ensures that progress towards reaching consensus may be made despite some byzantine nodes.

\begin{definition}[Quorum Intersection]
An FBAS enjoys \emph{quorum intersection} iff any two of its quorums share a node---i.e., for all quorums $U_1$ and $U_2$, $U_1 \cap U_2 \neq \varnothing$. (From~~\cite{mazieres2015stellar})
\end{definition}
\begin{definition}[Delete]
If $\langle V, Q \rangle$ is an FBAS and $B \subseteq V$ is a set of nodes, then to delete B from $\langle V, Q \rangle$, written $\langle V, Q \rangle^B$, means to compute the modified FBAS $\langle V \setminus B, Q^B \rangle$ where $Q^B(v) = \{ q \setminus B | q \in Q(v) \}$. (From~~\cite{mazieres2015stellar})
\end{definition}
\begin{definition}[DSet]
Let $\langle V, Q \rangle$ be an FBAS and $B \subseteq V$ be a set of nodes. We say $B$ is a dispensible set, or \emph{DSet}, iff:
\begin{enumerate}
    \item (quorum intersection despite $B$) $\langle V, Q \rangle^B$ enjoys quorum intersection.
    \item (quorum availability despite $B$) Either $V \setminus B$ is a quorum in $\langle V, Q \rangle$ or $B = V$.
\end{enumerate}
(From~~\cite{mazieres2015stellar})
\end{definition}

The concept of `dispensible set' (DSet) in an FBAS is key to understanding its security. The DSet contains a set of nodes that can act adversarially, without jeopardizing the safely and liveness properties of the FBAS for the remaining (honest) nodes. The `quorum intersection despite $B$'  property ensures safety, since it requires any two quorums within a system without $B$ nodes to intersect, and thus agree on the same result. The `liveness despite $B$ property' ensures the set of honest nodes in the FBAS can form a consensus to agree on a result, even if the nodes in $B$ do not participate.

In our experiments, to establish the security of \sysname as an FBAS system, we will need to determine the necessary DSet, containing byzantine nodes, as well as potentially some honest nodes for whom Sybil defences failed. However, establishing the `quorum intersection despite B' property of a DSet not easy: naively it would require computing all quorums and testing their pairwise intersection---which is computationally unfeasible for larger number of nodes. Therefore we devise two algorithms to determine the DSet is a quorum (property 2) and also to check quorum intersection despite the DSet (property 1), efficiently. The correctness of those algorithms depends on original theorems related to an FBAS, which may also be of independent interest. 

We first prove a lemma, on which we rely for the correctness of our algorithm.
\begin{lemma}\label{quoruminquorum}
Consider a node $v \in V$ in an FBAS $\langle V, Q \rangle$, and a quorum $U$, such that $v \in U$, and a quorum slice $q \in Q(v)$ for $v$ contained in $U$, namely $q \in U$. If another node $v'$ is in the same slice, namely $v' \in q$, and the minimum cardinally of any quorum that contains $v'$ is $h$---i.e.\ for all quorum $U'$, such that $v' \in U'$, $ h \leq |U'|$. Then the minimum cardinally of $U$ is also $h$---i.e. $h \leq |U|$.
\end{lemma}
\begin{proof}
The theorem seems complex, but really is the result of a simple symmetry: since the quorum slice $q$ is contained in $U$, both $v \in q$ and $v' \in q$ are within the quorum, $\{v, v'\} \subseteq U$. Since $h$ is the minimal cardinality of a quorum containing $v'$, and the quorum $U$ also contains $v'$, it trivially follows that $h \leq U$. 
\end{proof}

\begin{definition}[Quorum Slice Cardinality Map]
The function $C(v)$ is the \emph{quorum slice cardinality Map} for an an FBAS $\langle V, Q \rangle$. For each node $v \in V$ in , it returns the cardinality of the smallest quorum slice in $Q(v)$---i.e. $C(v) = \min\{|q|\, | \, q \in Q(v)  \}$.
\end{definition}

\begin{theorem}[Trivial Intersection]
 For an FBAS, if all values of the quorum slice cardinality map are larger than half the number of nodes, it enjoys quorum intersection---i.e. $\forall v \in V, C(v) > |V| / 2$.
\end{theorem}
\begin{proof}
Consider two nodes and quorums in the FBAS, $v_1 \in U_1$ and $v_2 \in U_2$. By the definition of quorums there must exist two slices $q_1 \in Q(v_1), q_2 \in Q(v_2)$ of $v_1, v_2$ respectively such that $q_1 \subseteq U_1$ and $q_2 \subseteq U_2$. Since $|q_1| > |V| / 2$ and $|q_2| > |V| / 2$, it must be that they intersect at least in one element. And $q_1 \cup q_2 \neq \varnothing \Leftrightarrow q_1 \cap U_1 \cap q_2 \cap U_2 \neq \varnothing \Rightarrow U_1 \cap U2 \neq \varnothing$.
\end{proof}

This first theorem provides a trivial way to check a FBAS for quorum intersection: if all quorum slices, for all nodes contain more than half the nodes, then all quorums will intersect. However, this condition is much stronger than necessary for quorum intersection, and in practice not always achievable. However, it is the property on which the security of traditional BFT systems can be proven in when those are encoded as FBAS. However, we will seek a weaker property that still implies quorum intersection.


\begin{lemma}[Minimum Quorum Cardinality in FBAS]
Consider the  FBAS $\langle V, Q \rangle$. We define a function $F_i(v)$ providing a lower bound on the cardinality of any quorum $U$ containing $v$, namely $\forall U,i$ such that $v \in U$ it holds that $F_i(v) \leq |U|$. We initialize $F$ as $F_0(v) = C(v)$, where $C$ is the quorum slice cardinality map. We also define the sets $\bar{q}_v \subseteq$ for each $v \in V$, containing all nodes in $v$'s quorum slices---i.e.\ $\bar{q}_v = \bigcup_{q \in Q(v)} q$.

Define as $S$ the sequence of values $[F_i(v')\, |\, v' \in \bar{q}_v]$ in ascending order, and $S_i[C(v)]$ is its $C(v)^{th}$ element. If we assign  $F_{i+1}(v) \leftarrow \max\{S_i[C(v)],\, F_{i}(v)\}$, the value $F_{i+1}(v)$ is also a lower bound on the cardinally of any quorum containing $v$. 
\end{lemma}
\begin{proof}
By the definition of the quorum slice cardinality map $C$ it is trivial to argue that $F_0(v) = C(v)$ is a lower bound on the cardinality of all quorums including $v$, since they each need to fully contain at least one quorum slice from $v$. 
We need to show that the value $F_i(v)$ and therefor $S_i[C(v)]$ is a lower bound on the cardinality of any quorum $U$ containing $v$. Since the minimum quorum slice size of $v$ is $C(v)$ it must contain at least that number of nodes, out of the set $\bar{q}_v$. By lemma~\ref{quoruminquorum} we know that including a node $v'$ from $\bar{q}$ into a quorum $U$, would yield a quorum of cardinality at least $F_i(v')$. Since $C[v]$ such nodes from $\bar{q}_v$ must be included the minimum cardinality of $U$ is $S_i[C(v)]$, since $S_i$ is defines as the ordered sequence of minimum cardinality sizes for quorums each node in $\bar{q}_v$. 
\end{proof}
We use this lemma in building an algorithm that computes lower bounds on the cardinalities quorums of all nodes in the FBAS iterativelly. It starts with an estimation equal to the cardinality of the smallest quorum slice for each node (the function L(v)), and then increases the estimate as $F_{i+1}(v) \leftarrow \max( F_i(v),\, S_i[L(v)] )$, where $S_i = [F_i(v') \, | \, v' \in \bar{q}_v]$.

Based on the above we can efficiently estimate minimum bounds on the cardinality of all quorums in the \sysname FBAS. Then we can test them to show quorum intersection:
\begin{theorem}[Quorum Intersection due to Minimum Size]
If within a FBAS, all quorums $U$ have cardinality $|U| > |V| / 2$, it enjoys quorum intersection.
\end{theorem}
\begin{proof}
Trivially, if we have two sets $U_1 \subseteq V$ and $U_2 \subseteq V$ with a number of elements greater than half the number of elements in $V$, they must intersect in at least one element. 
\end{proof}

We note that Quorum Intersection due to Minimum Size is a sufficient condition to guarantee a FBAS enjoys quorum intersection, but it is too strong to be necessary. For example an FBAS with a dictator node $v_0$ present in all quorums, will satisfy trivially quorum intersection, without the need for all quorums to be of a certain size.

%

\para{Computing DSets and FBAS safety} We leverage the theorems above to test the concrete FBAS extracted from \sysname for safety and liveness. We define a set of nodes $V$, each with a list $H(v) \subseteq V$ of nodes they consider honest, and a set of malicious nodes $B$. The quorum function $Q(v)$ contains all subsets of $H(v)$ of size over $2/3 |H(v)|$. 

First we execute a procedure \textsc{DetermineDSet} using the initial bad nodes $B$ to determine a set $B'$ of nodes that cannot reach agreement, due to having accepted too many bad nodes as honest. Following the terminology from Stellar we call the set of nodes $B' \setminus B$ befouled nodes. A node is befouled if it has accepted in its set $H(v)$ more than a third of bad or befouled nodes. By definition the nodes $V \setminus B'$ still constitute a quorum. (Thus proving liveness despite $B'$).

Once we have identified the set of bad and befouled nodes $B'$, we define the FBAS $\langle V, Q \rangle^{B'}$ and try to establish whether it enjoys quorum intersection. We use algorithm \textsc{DetermineSafety} to test for quorum intersection: we iteratively determine an increasingly better lower bound $F_i(v)$ on the quorum cardinality of each node. Once the bound converges, we check that each $F_i(v) > |V \setminus B'| / 2$, which according to our theorems ensures quorum intersection (Thus proving quorum intersection despite $B'$). If this condition is true we label our FBAS as safe, since the set $B'$ is a DSet.

\begin{figure*}[t]
\begin{subfigure}[t]{0.50\textwidth}
\begin{algorithmic}
\Function{DetermineDSet}{$\langle V, Q \rangle, B$}
    \State $H(v) \leftarrow \bigcup_{q \in Q(v)} q$ 
    \State exit $\leftarrow$ False
    \While{\textbf{not} exit}
        \State exit $\leftarrow$ True
        \ForAll{ $V \setminus B$ }
        \If{$|H(v) \setminus B| > 2 \cdot | H(v) \cap B |$}
            \State $B \leftarrow B \cup \{ v \}$
            \State exit $\leftarrow$ False
        \EndIf
        \EndFor
    \EndWhile
    \State \Return $B$
\EndFunction
\end{algorithmic}
\end{subfigure}
\begin{subfigure}[t]{0.50\textwidth}
\begin{algorithmic}
\Function{DetermineSafety}{$\langle V, Q \rangle$}
\State $H(v) \leftarrow \bigcup_{q \in Q(v)} q$ 
\State $C(v) \leftarrow \min \{ |q| \text{ for } q \in Q(v)  \}$
\State $i \leftarrow 0$
\State $F_i(v) \leftarrow C(v)$
 \State exit $\leftarrow$ False
    \While{\textbf{not} exit}
        \State exit $\leftarrow$ True
        \ForAll{ $V \setminus B$ }
            \State $S \leftarrow \text{sorted}([F_i(v') \text{ for } v' \in H(v)])$
            \State $F_{i+1}(v) \leftarrow \max\{ F_i(v),\, S_i[C(v)] \}$
            \If{ $ F_{i+1}(v) \neq F_{i}(v)$}
            \State exit $\leftarrow$ False
            \EndIf
        \EndFor
        \State $i \leftarrow i + 1$
    \EndWhile
    \State \Return $\forall v \in V. F_{i}(v) > |V| / 2$
\EndFunction
\end{algorithmic}
\end{subfigure}
\caption{Algorithms to determine safe set and quorum intersection}
\end{figure*}


\section{Experimental Evaluation}

\para{Datasets \& Pre-processing} An evaluation of \sysname necessitates the use of real-world datasets of social connections, since the mechanism relies on the dynamics of connections within `real' social networks. There are methodological challenges to doing this. First, there does not exist a network embodying the proposed mechanism of establishing links backed by mutual token on links as necessary by \sysnameh. Second, existing datasets are based on networks for casual socializing, in which incentives are not aligned for careful selection of links, but rather provide advantages and incentives for users to be prosmicuous in their connections. Those issues present threats to validity.

For our evaluation we chose to use the \emph{pokec} network dataset, that is open and available on the Stanford large  network dataset collection\footnote{\url{https://snap.stanford.edu/data/soc-Pokec.html}}. This is a snapshop of the largest social network provider in Slovakia, collected in 2012. It contains 1632803 nodes and 30622564 edges. 

We pre-process this network in two ways, to produce evaluation datasets: (1) We sub-sample 200000 nodes from the network, and create a subgraph with all their edges (including those to nodes not in the set of nodes); (2) We then recursively prune the network to a core of nodes with degree at least 3. Pruning is performed by removing nodes with degree less than 3, until all nodes have a higher degree. 

These operations result in sub-graphs of size about 10000 nodes, which is comparable to the number of miners in systems such as Bitcoin and Ethereum. We extract the degree three core of the network as a proxy for nodes that have strong connections to each other, excluding nodes with weaker trust connections between them. (Note such pre-processing can also done as part of a production \sysname pipeline.)

\para{Sybil Attack Simulation} To evaluate \sysname we simulate Sybil attacks on the graph datasets, in a what that is most generic. We parametrize the attack through a number of parameters: (1) the number of Sybil nodes ($n_s$); (2) the number of links or amount of stake on links purely in the Sybil region ($l_s$); (3) the amount of stake on links between the honest region and Sybil region ($l_n$); (4) the fraction of honest nodes that are naive, and connect to Sybil nodes ($f_n$).

Given those parameters, we instanciate a set of Sybil nodes of size $n_s$, and establish $l_s$ mutual links between them at random. We sample at random a set of honest nodes to be `naive', as a fraction $f_n$ of the honest nodes. We then create mutual connections between random Sybil nodes in that set, and honest nodes from the naive set, according to the budget of links or stake available ($l_n$). 

SybilInfer~\cite{DBLP:conf/ndss/DanezisM09} provides an argument that the exact composition of the Sybil region does not impact security, but what matters is rather the relative size of the Sybil region and the links between honest and Sybil regions. However, we there might be optimizations in connecting Sybils in specific ways to the naive nodes, which we have not explored. This is a further threat to the validity of our results. However, our methodology is in line with previous works.

\para{Purely Benign or Byzantine Conditions} We first evaluate the \sysname mechanism in a network composed of overwhelmingly benign nodes. We instanciate such a network by only attaching a single Sybil ($n_ = 1$), no stake in the Sybil region ($l_s = 0$), and minimal stake between the honest and dishonest nodes ($l_n = 2$). On the other hand we allow this single Sybil nodes to connect to any honest node ($f_n = 1.0$). We primarily use this condition to ensure that \sysname does not suffer from false positives in detecting Sybils, that compromise agreement, as the raw SybilInfer mechanism does.

We also evaluate \sysname under conditions of byzantine attacks that can be accommodated within the traditional Byzantine fault tolerance paradigm, with just a standard proof-of-stake system: where the number of Sybils $n_s$ is $1/3$ of the size of the honest nodes, and the stake of all links connected to the Sybil nodes is at most $1/2$ of the total stake in the honest region. We also allow Sybil nodes to connect to any honest node ($f_n = 1.0$). This condition simulates an attack that can be tolerated even if all Sybil nodes are accepted as honest by all -- but we need to assess whether this is the case in \sysname and whether the resulting FBAS is secure.

In both conditions \sysname honest nodes reach safe agreement. In the benign condition the cut-off is determined as $y = 0.49$, and all nodes are accepted by all other honest nodes as honest. This includes the single Sybil node. All quorums are larger than half the number of honest nodes, and global agreement is reached. In the byzantine condition agreement is also reached. The cut-off value is set automatically as $y = 0.50$. And since all honest nodes have fewer than 1/3 links to Sybil nodes (none is confused) they reach agreement. (Those are the results of 10 repeats of the experiments, for different configurations of the Sybil attack). 

Those results confirm that \sysname performs no worse than not using \sysname --- which is not a given: (1) when there is no Sybil attack it does not impede agreement through false positives; and (2) when there are fewer than $1/3$ dishonest nodes, the FBAS is secure and preserves agreement and liveness without any negative interference from \sysname.

\para{Future Evaluation} The analysis above, concerning benign and byzantine settings will be extended in the full paper to establish the security of \sysname for:
\begin{itemize}
\item{Variable Number of Adversarial Nodes}
\item{Variable Size of Adversarial Stake}
\item{Variable Numbers of Confused Honest Nodes}
\item{Variable Size of Honest-Sybil Links}
\end{itemize}

\section{Conclusions}

In this work we show that we can leverage social networks to protect traditional Proof-of-Stake systems against a wider range of attacks, from nodes with a lot of stake, than previously expected. To do so it is necessary to abstract their consensus mechanisms within the more general FBAS framework, and also devise efficient tests for whether such an FBAS is secure to support experimental evaluations. The degree to which this mechanism is effective is subject to extended evaluation, which will be ready for the Stanford Blockchain Conference.


{\small \bibliographystyle{alpha}
\bibliography{references}}

\newcommand{\etalchar}[1]{$^{#1}$}
\begin{thebibliography}{AKM{\etalchar{+}}15}

\bibitem[ABB{\etalchar{+}}18]{DBLP:conf/eurosys/AndroulakiBBCCC18}
Elli Androulaki, Artem Barger, Vita Bortnikov, Christian Cachin, Konstantinos
  Christidis, Angelo~De Caro, David Enyeart, Christopher Ferris, Gennady
  Laventman, Yacov Manevich, Srinivasan Muralidharan, Chet Murthy, Binh Nguyen,
  Manish Sethi, Gari Singh, Keith Smith, Alessandro Sorniotti, Chrysoula
  Stathakopoulou, Marko Vukolic, Sharon~Weed Cocco, and Jason Yellick.
\newblock Hyperledger fabric: a distributed operating system for permissioned
  blockchains.
\newblock In {\em Proceedings of the Thirteenth EuroSys Conference, EuroSys
  2018, Porto, Portugal, April 23-26, 2018}, pages 30:1--30:15, 2018.

\bibitem[AKM{\etalchar{+}}15]{armknecht2015ripple}
Frederik Armknecht, Ghassan~O Karame, Avikarsha Mandal, Franck Youssef, and
  Erik Zenner.
\newblock Ripple: Overview and outlook.
\newblock In {\em International Conference on Trust and Trustworthy Computing},
  pages 163--180. Springer, 2015.

\bibitem[B{\etalchar{+}}02]{back2002hashcash}
Adam Back et~al.
\newblock Hashcash-a denial of service counter-measure, 2002.

\bibitem[Dai98]{dai1998b}
Wei Dai.
\newblock b-money, 1998.
\newblock {\em URL: http://www. weidai. com/bmoney. txt}, 1998.

\bibitem[DH18]{danezis2018blockmania}
George Danezis and David Hrycyszyn.
\newblock Blockmania: from block dags to consensus.
\newblock {\em arXiv preprint arXiv:1809.01620}, 2018.

\bibitem[DLKA05]{DBLP:conf/esorics/DanezisLKA05}
George Danezis, Chris Lesniewski{-}Laas, M.~Frans Kaashoek, and Ross~J.
  Anderson.
\newblock Sybil-resistant {DHT} routing.
\newblock In {\em Computer Security - {ESORICS} 2005, 10th European Symposium
  on Research in Computer Security, Milan, Italy, September 12-14, 2005,
  Proceedings}, pages 305--318, 2005.

\bibitem[DM09]{DBLP:conf/ndss/DanezisM09}
George Danezis and Prateek Mittal.
\newblock Sybilinfer: Detecting sybil nodes using social networks.
\newblock In {\em Proceedings of the Network and Distributed System Security
  Symposium, {NDSS} 2009, San Diego, California, USA, 8th February - 11th
  February 2009}, 2009.

\bibitem[Dou02]{DBLP:conf/iptps/Douceur02}
John~R. Douceur.
\newblock The sybil attack.
\newblock In {\em Peer-to-Peer Systems, First International Workshop, {IPTPS}
  2002, Cambridge, MA, USA, March 7-8, 2002, Revised Papers}, pages 251--260,
  2002.

\bibitem[KRDO17]{kiayias2017ouroboros}
Aggelos Kiayias, Alexander Russell, Bernardo David, and Roman Oliynykov.
\newblock Ouroboros: A provably secure proof-of-stake blockchain protocol.
\newblock In {\em Annual International Cryptology Conference}, pages 357--388.
  Springer, 2017.

\bibitem[Kwo14]{kwon2014tendermint}
Jae Kwon.
\newblock Tendermint: Consensus without mining.
\newblock {\em Draft v. 0.6, fall}, 2014.

\bibitem[LC04]{laurie2004proof}
Ben Laurie and Richard Clayton.
\newblock Proof-of-work proves not to work; version 0.2.
\newblock In {\em Workshop on Economics and Information, Security}, 2004.

\bibitem[Les18]{samlessin}
Sam Lessin.
\newblock Venmo trust and the blockchain.
\newblock {\em
  \url{https://www.theinformation.com/articles/venmo-trust-and-the-blockchain}},
  2018.

\bibitem[Lev09]{levien2009attack}
Raph Levien.
\newblock Attack-resistant trust metrics.
\newblock In {\em Computing with Social Trust}, pages 121--132. Springer, 2009.

\bibitem[Maz15]{mazieres2015stellar}
David Mazieres.
\newblock The stellar consensus protocol: A federated model for internet-level
  consensus.
\newblock {\em Stellar Development Foundation}, 2015.

\bibitem[MTHK12]{DBLP:conf/ccs/MohaisenTHK12}
Abedelaziz Mohaisen, Huy Tran, Nicholas Hopper, and Yongdae Kim.
\newblock On the mixing time of directed social graphs and security
  implications.
\newblock In {\em 7th {ACM} Symposium on Information, Compuer and
  Communications Security, {ASIACCS} '12, Seoul, Korea, May 2-4, 2012}, pages
  36--37, 2012.

\bibitem[Nak08]{nakamoto2008bitcoin}
Satoshi Nakamoto.
\newblock Bitcoin: A peer-to-peer electronic cash system.
\newblock 2008.

\bibitem[Sch90]{schneider1990implementing}
Fred~B Schneider.
\newblock Implementing fault-tolerant services using the state machine
  approach: A tutorial.
\newblock {\em ACM Computing Surveys (CSUR)}, 22(4):299--319, 1990.

\bibitem[Woo14]{wood2014ethereum}
Gavin Wood.
\newblock Ethereum: A secure decentralised generalised transaction ledger.
\newblock {\em Ethereum project yellow paper}, 151:1--32, 2014.

\bibitem[YGKX10]{DBLP:journals/ton/YuGKX10}
Haifeng Yu, Phillip~B. Gibbons, Michael Kaminsky, and Feng Xiao.
\newblock Sybillimit: {A} near-optimal social network defense against sybil
  attacks.
\newblock {\em {IEEE/ACM} Trans. Netw.}, 18(3):885--898, 2010.

\bibitem[YKGF08]{DBLP:journals/ton/YuKGF08}
Haifeng Yu, Michael Kaminsky, Phillip~B. Gibbons, and Abraham~D. Flaxman.
\newblock Sybilguard: defending against sybil attacks via social networks.
\newblock {\em {IEEE/ACM} Trans. Netw.}, 16(3):576--589, 2008.

\end{thebibliography}

\end{document}